\documentclass[conference,letterpaper]{article}
\usepackage[colorinlistoftodos,prependcaption,textsize=tiny]{todonotes}

\usepackage{ifpdf}
\usepackage{cite}
\usepackage{graphicx}
\usepackage{amssymb}
\usepackage{amsmath,amsfonts}
\usepackage{stmaryrd} 
\usepackage{xcolor}
\usepackage{algorithmic}
\usepackage{textcomp}
\usepackage{bbm}
\usepackage{stmaryrd} 
\usepackage{amsthm}
\usepackage{mathrsfs}
\usepackage{tikz}
\usepackage{setspace}
\usetikzlibrary{automata,arrows,positioning,calc}
\definecolor{celadon}{rgb}{0.67, 0.82, 0.59}
\definecolor{cblue}{rgb}{0.6, 0.73, 0.89}
\theoremstyle{remark}
\newtheorem{thm}{Theorem}

\newtheorem{lem}{Lemma}
\newtheorem{defi}{Definition}
\newtheorem{ex}{Example}

\begin{document}
		\begin{spacing}{0.95}
	
	\title{An Explicit Formula for the\\Zero-Error Feedback Capacity of a \\Class of Finite-State Additive Noise Channels}
		\author{Amir Saberi\textsuperscript{ a}, Farhad Farokhi\textsuperscript{ a,b} and Girish N.~Nair\textsuperscript{ a}}
	\author{{Amir Saberi, Farhad Farokhi and Girish N.~Nair}
	\thanks{The authors are with the Department of Electrical and Electronic Engineering, University of Melbourne, VIC 3010, Australia  (e-mails: asaberi@student.unimelb.edu.au, \{ffarokhi, gnair\}@unimelb.edu.au). F. Farokhi is also with the CSIRO's Data61 (e-mail: farhad.farokhi @data61.csiro.au). $ \qquad  \qquad \qquad \qquad \qquad \qquad \qquad \qquad \qquad \qquad \qquad$
		{\textcopyright} 2020. This manuscript version is made available under the CC-BY-NC-ND 4.0 license http://creativecommons.org/licenses/by-nc-nd/4.0/ }}

	\maketitle
	\thispagestyle{empty}
	
	\begin{abstract}
		It is known that for a discrete channel with correlated additive noise, the ordinary capacity with or without feedback both equal $ \log q-\mathcal{H} (Z) $, where $ \mathcal{H}(Z) $ is the entropy rate of the noise process $ Z $ and $ q $ is the alphabet size. In this paper, a class of finite-state additive noise channels is introduced. It is shown that the zero-error feedback capacity of such channels is either zero or $C_{0f} =\log q -h (Z) $, where $ h (Z) $ is the {\em topological entropy} of the noise process. A topological condition is given when the zero-error capacity is zero, with or without feedback. Moreover, the zero-error capacity without feedback is lower-bounded by $ \log q-2 h (Z) $. We explicitly compute the zero-error feedback capacity for several examples, including channels with isolated errors and a Gilbert-Elliot channel.
	\end{abstract}

	\section{Introduction}
	In his 1956 paper \cite{shannon1956zero}, Shannon introduced the concept of zero-error communication. Although, a general formula is still missing for the zero-error capacity $ C_0 $ of a discrete memoryless channel (DMC) without feedback, Shannon derived one for the zero-error capacity $ C_{0f} $ of a DMC with noiseless feedback. In recent years, there has been progress towards determining $ C_{0f} $ for channels with memory. In \cite{zhao2010zero}, Zhao and Permuter introduced a dynamic programming formulation for computing  $ C_{0f} $ for a finite-state channel modeled as a Markov decision process, assuming state information is available at both encoder and decoder. However, the problem is still open when there is no state information at the decoder.
	
	In this paper, we study the zero-error capacity, with and without feedback, of discrete channels with additive correlated noise. The ordinary capacities with and without feedback of such channels are studied in \cite{alajaji1995feedback}, in which it is proved that
	 \begin{align}
	 C=C_f=\log_2 q-\mathcal{H} (Z),\label{ccf}
	 \end{align}
	where $q$ is the input alphabet size and $\mathcal{H}(Z) $ is the entropy rate of the noise process $ Z $. In this paper, we consider additive noise channels where the noise is generated by a finite-state machine.  We prove a similar formula for the zero-error feedback capacity $ C_{0f} $ and a lower bound for the zero-error capacity $ C_0 $, in terms of {\em topological entropy} (Theorem \ref{thm:zero-error}).  Unlike \cite{zhao2010zero}, we do not assume that channel state information is available at the encoder or decoder. In \cite{saberi2019state}, we studied $ C_0 $ of some special cases of these channels and derived a similar lower bound. In this paper, we extend that result to a more general channel model, and also derive an exact formula for $C_{0f}$.  Examples including the well-known Gilbert-Elliot channel are considered, for which the explicit value of $ C_{0f} $ is computed. To the best of our knowledge, this has not been done for these channels.
	
	The rest of paper is organized as follows. In Section~\ref{sec:model} the channel model and main results are presented. Proofs are given in Sections~\ref{sec:proof} and \ref{sec:proofl} and some examples are discussed in \ref{sec:examps}. Finally, concluding remarks and future extensions are discussed in section~\ref{sec:conclusions}.
	
	Throughout the paper, calligraphic letters such as $ \mathcal{X} $, denote sets. The cardinality of set $ \mathcal{X} $ is denoted by $ |\mathcal{X}| $. The channel input alphabet size is $ q $, logarithms are in base $2$. Random variables are denoted by upper case letters such as $ X $, and their realizations by lower case letters such as $ x $. The vector $ (x_i)_{i=m}^n $ is denoted by $ x_{m:n} $.

	\section{Channel Model and Main Results} \label{sec:model}
	Let the input, output and noise at time $ i\in \mathbb{N} $ in the channel be $ x_i \in \mathcal{X} $, $ y_i \in \mathcal{Y} $, and $ z_i \in \mathcal{Z} $, respectively. Before we describe the channel, we define the following notion. 	
	\begin{defi}[Finite-state machine] \label{def:fsm}
		A {\em finite-state machine} is defined as directed graph $\mathscr{G}=(\mathcal{S},\mathcal{E})$, where the vertex set $ \mathcal{S}=\{0,1,\dots,|\mathcal{S}|-1\} $ denotes {\em states} of the machine, and the edge set $\mathcal{E}\subseteq\mathcal{S}\times \mathcal{S}$ denotes possible transitions between two states. We say a process $(S_i)_{i\geq 1}$ is described by $\mathscr{G}$ if a) there is a positive probability that any state is eventually visited, i.e. $\forall s\in\mathcal{S}$, $\exists i\geq 1$ s.t.  $P(S_i=s)>0$, b) if $(s,s')\in\mathcal{E}$, a transition  $s\to s'$ is always possible for all possible past state sequences, i.e. $P(S_{i+1}=s'|S_i=s, s_{1 :i-1}) > 0$ whenever $P(S_i=s, s_{1 :i-1}) > 0$; and c) conversely, if $(s,s')\notin \mathcal{E}$, then $P(S_{i+1}=s'|S_i=s, s_{1 :i-1}) = 0$ whenever $P(S_i=s, s_{1 :i-1}) > 0$.
	\end{defi}
	
	{\bf Remark:} Processes described by a finite-state machine are {\em topologically Markov} [5, Ch.2], but need not be stochastic Markov chains.
	
	The following channel is studied in this paper.
	\begin{defi}[Finite-state additive noise channels]
		\label{def:fanchannel}	
		A discrete channel with common input, noise and output $ q $-ary alphabet $ \mathcal{X} $ is called {\em finite-state additive noise} if its output at time $ i \in \mathbb{N} $ is obtained by
		\begin{align*}
		Y_i = X_i \oplus Z_i, \, i \in \mathbb{N},
		\end{align*}
		where $ \oplus $ is modulo $ q $ addition and the correlated additive noise $ Z_i $ is governed by a state process $ (S_i) $ on a finite-state machine such that each outgoing edge from a state $ s_i $ corresponds to different values $ z_i $ of the noise. Thus, there are at most $ q $ outgoing edges from each state.
		We assume the state transition diagram of the channel is strongly connected  and that $Z_i$ is independent\footnote{ This can be relaxed to {\em qualitative independence} \cite[Ch.1]{renyi1970foundations}; i.e. if $P(X_{1:k}~=~x_{1:k})$ and $P(Z_k=z_k)$ are both positive, then  $P(X_{1:k}=x_{1:k},Z_k=z_k)>0$.} of $X_{1:i}$.
	\end{defi}
	Figure \ref{fig:examp} shows a noise process which defines a channel that has no more than two consecutive errors. For example, the transition at time $ i $ from state $ S_i=0 $ to itself corresponds to $ Z_i=0 $. Moreover, $ Z_i=1 $ leads to the transition ending in state $ S_{i+1}=1 $ (state at next time step). Note that, in $ S_i=2 $, the noise can only take $ Z_i=0 $ and transits to $ S_{i+1}=0 $. 
	
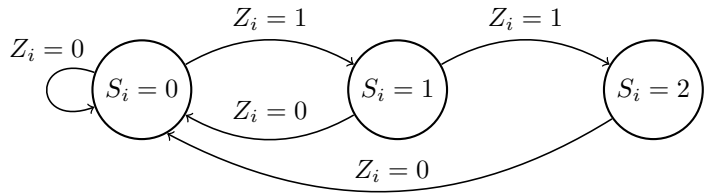
\begin{figure}[t]
	\centering
	\begin{tikzpicture}
	[->, auto, semithick, node distance=3.4cm]
	\tikzstyle{every state}=[fill=white,draw=black,thick,text=black,scale=1]
	\node[state]   (S1)       {$ S_i=0 $};
	\node[state]   (S2)[right of=S1]   {$ S_i=1 $};
	\node[state]   (S3)[right of=S2]   {$ S_i=2 $};
	\path
	(S1) edge[out=160, in=200,looseness=5] node [above=0.25] {$ Z_i=0 $}(S1)
	edge[bend left] node [above] {$ Z_i=1 $}  (S2)
	(S2) edge[bend left]  node [above=0.1] {$ Z_i=0 $}  (S1)
	edge[bend left]  node [above] {$ Z_i=1 $}  (S3)
	(S3) edge[bend left]  node [above] {$ Z_i=0 $}  (S1.300);
\end{tikzpicture}
\caption{State transition diagram of a noise process in a channel at which no more than two consecutive errors can happen in the channel.}
\label{fig:examp}
\end{figure}
\begin{defi}[Coupled graph] \label{def:dg}
	{\em Coupled graph} of a finite-state machine (with labeled graph $ \mathscr{G} $) is defined as a labeled directed graph\footnote{This product is called {\em tensor product},  as well as {\em Kronecker  product} \cite[Ch. 4]{hammack2011handbook}. } $\mathscr{G}_c=\mathscr{G}\times \mathscr{G}$, such that it has vertex set $ V=\mathcal{S}\times \mathcal{S} $ and has an edge from node $ u=(i,j) \in V $ to $ v=(k,m) \in V  $ if and only if there are edges from $ S=i $ to $ S=k $ (with a label value $ E_{ik} $) and from $ S=j $ to $ S=m $ (with a label value $ E_{jm} $) in $ \mathscr{G} $, each edge has a label equal to $E_{ik} \ominus E_{jm}$, where $ \ominus $ is modulo $ q $ subtraction.
\end{defi}
For a state-dependent channel, the zero-error capacity is defined as follows.
\begin{defi} \label{def:c0}
	The zero-error capacity, $C_0$, is the largest block-coding rate that permits zero decoding errors, i.e., 
	\begin{align}
	C_0:=\sup_{n \in \mathbb{N}, \, \mathcal{F} \in \mathscr{F}}
	\frac{\log |\mathcal{F}|}{n}, \label{c0def}
	\end{align}
	where $ \mathscr{F}\subseteq \mathcal{X}^{n}$  is the set of all block codes of length $ n $ that yield zero decoding errors for any channel noise sequence and channel initial state, such that no state information is available at the encoder and decoder. In a zero-error code, any two distinct codewords $ x_{1:n},\, x'_{1:n} \in \mathscr{F} $ can never result in the same channel output sequence, regardless of the channel noise and initial state.
\end{defi}	
The zero-error feedback capacity $C_{0f}$ is defined in the presence of a noiseless feedback from the output. In other words, assuming $ m \in \mathcal{M} $ is the message to be sent and  $ y_{1:n} $ is the output sequence received then $ x_i(m)=f_{m,i}(y_{1:i-1}), \, i=1,\dots, n, $ where $ f_{m,i} $ is the encoding function. Let the family of encoding functions $ \mathcal{F_{\mathcal{M}}}=\{f_{m,n}: m \in \mathcal{M}\} $. The zero-error feedback capacity, is the largest block-coding rate that permits zero decoding errors.

Before, presenting the main results, we need some preliminaries from symbolic dynamics.
In symbolic dynamics, topological entropy is defined as the asymptotic growth rate of the number of possible  state sequences. For a finite-state machine with an irreducible transition matrix $\mathcal{A}$, the topological entropy $ h $ is known to coincide with $\log\lambda$, where $\lambda$ is the {\em Perron value} of $\mathcal{A}$~\cite{lind1995introduction}. This is essentially due to the fact that the number of the paths from state $ S=i $ to $ S=j $ in $ n $ steps is the $ (i+1,j+1) $-th element of $ \mathcal{A}^n$, which grows at the rate of $\lambda^n$ for large $n$.

First we give a topological condition on when zero-error capacity is zero, with or without feedback.
\begin{thm}\label{thm:c0f0}
	The zero-error capacity with(out) feedback $ C_{0f} \,$(resp. $ C_0 $) of a finite-state additive noise channel [Def. \ref{def:fanchannel}] having finite-state machine [Def. \ref{def:fsm}] graph  $\mathscr{G}=(\mathcal{S},\mathcal{E}) $  is zero, if and only if $ \forall \, d_{1:n} \in \mathcal{X}^{n} , n \in \mathbb{N} $, there exists a walk on the {\em coupled graph} [Def. \ref{def:dg}] of $ \mathscr{G} $ with the label sequence $ d_{1:n} $. 
\end{thm}
\textbf{Remark:} This result implies that $ C_0 = 0 $ if and only if $ C_{0f} = 0 $ for finite-state additive noise channels.
\begin{proof}	
	\textit{Sufficiency:} We show that for any choice of encoding functions and blocklength $ n $ there is a common output for $ m,m' \in \mathcal{M}$, i.e., $ \exists z_{1:n},z'_{1:n} $ such that the output sequences, $ y_{1:n}=y'_{1:n} $, where $ y_{1:n}=f_{m,1:n}\oplus z_{1:n}, \,y'_{1:n}=f_{m',1:n}\oplus z'_{1:n} $. In other words, $\forall \, n \in \mathbb{N}$, and  
\[ d_{1:n}:=f_{m',1:n}(z'_{1:n-1}) \ominus f_{m,1:n}(z_{1:n-1})\, \in \mathcal{X}^n,\]
$\exists\, z_{1:n},z'_{1:n}$ such that $d_{1:n}(z_{1:n-1},z_{1:n-1})= z_{1:n} \ominus z'_{1:n} $.

First observe that having current states $ S_i=s$ and $S'_i=s' $, for two noise sequences of $ z_{1:i-1}$ and $z'_{1:i-1} $, respectively, the label on out-going edges in the coupled graph is belong to $ \{z_i \ominus z'_i |  S_i=s,  S'_i=s' \}$. 
Now consider the first transmission, by choosing any inputs $ f_{m,1}, f_{m',1} \in \mathcal{X}$, if there is an edge from any state $ (k,j) \in V $ with the value $d_1:=f_{m',1}\ominus f_{m,1} \in \mathcal{X}$ then there exist $ z_1,z'_1 \in \mathcal{X} $ that produce a common output for two channel inputs $ f_{m,1} $ and $ f_{m',1} $. By continuing this argument for any $ i\in \mathbb{N} $ having $ y_{1:i-1}=y'_{1:i-1} $, if $ d_i =f_{m',i}(z'_{1:i-1}) \ominus f_{m,i}(z_{1:i-1}) \in \mathcal{X} $ is chosen such that there is an edge with value $ d_i $ then there is an output shared with two messages. In other words, by choosing any value for $ d_i $, if there is an edge with corresponding value it means there is a pair of noise values $ (z_i,z'_i) $ such that $ d_i=z_i \ominus z'_i $, therefore $ y_i=y'_i $. If there is no such an edge for a particular $ d_i $, then there is no pair of noise values that produces the same output, and thus, $  y_i \neq y'_i $.

Therefore, if $ \forall n \in \mathbb{N} $ and for any choice of $ d_{1:n} \in \mathcal{X}^n $ there is a walk on the coupled graph then the corresponding noise sequences of the walk can produce the same output, i.e. $ y_{1:n}=y'_{1:n} $ which implies $ C_{0f}=0 $ and therefore $ C_0=0 $.\\

\textit{Necessity:} Assume there is no walk for a sequence of $ d_{1:n} $ then by choosing any two input sequences $x_{1:n}, x'_{1:n}$ such that $ x_{1:n}\ominus x'_{1:n}=d_{1:n} $, two messages $ m $ and $ m' $ can be transmitted with zero-error which contradict with the assumption that $ C_{0}=0 $ (and also $ C_{0f}=0 $).
\end{proof}
We now relate the zero-error capacities of the channel to the noise process topological entropy.

\begin{thm}\label{thm:zero-error}
	The zero-error feedback capacity of the finite-state additive noise channel [Def. \ref{def:fanchannel}] with topological entropy $h(Z)$ of the noise process where no state information is available at the transmitter and decoder is either zero or 
	\begin{align}
	C_{0f} &= \log q-h(Z).\label{c0f}
	\end{align}
	Moreover, the zero-error capacity (without feedback) is lower bounded by
		\begin{align}
	C_{0} &\geq \log q - 2h(Z).\label{c0l}
	\end{align}
\end{thm}
\textbf{Remarks:} \begin{itemize} \item The zero-error feedback capacity has a similar representation to the ordinary feedback capacity in \eqref{ccf} but with the stochastic noise entropy rate  $ \mathcal{H} (Z) $ replaced with the topological entropy $h(Z)$.
	
	\item The topological entropy can be viewed as the rate at which the noise dynamics generate uncertainty. Intuitively, this uncertainty cannot increase which explains why it appears as a negative term on the right hand side of \eqref{c0f} and \eqref{c0l}. Moreover, the sum of zero-error feedback capacity and the topological entropy is always equal to $ \log q $, meaning that if the noise uncertainty is increased, the same amount will be decreased in the capacity.
	
	\item The result of \eqref{c0f} is an explicit closed-form solution, which is a notable departure from the iterative, dynamic programming solution in \cite{zhao2010zero}.
	
	\item Following Definition \ref{def:fanchannel}, the channel states are not assumed to be Markov, just topologically Markov. Thus the  transition probabilities in the finite-state machine can be time-varying dependent on previous states. In other words, as long as the graphical structure is not changed, the result is valid.
\end{itemize}

\section{Proof of the Zero-error Feedback Capacity} \label{sec:proof}

The conditions on when $ C_{0f}=0 $ is given in Theorem \ref{thm:c0f0}. Here, we consider $ C_{0f}>0 $. Before presenting the rest of the proof, we give the following lemma.
\begin{lem}\label{lem:out}
	For a finite-state additive noise channel with irreducible adjacency matrix, there exist positive constants $ \alpha $ and  $ \beta $ such that, for any input sequence $ x_{1:n} \in \mathcal{X}^n $, the number of all possible outputs 
	\begin{align}
	\alpha \lambda^n \leq |\mathcal{Y}(s_0,x_{1:n}) |=|\mathcal{Z}(s_0,n) | \leq  \beta  \lambda^n , \label{outlam}
	\end{align}
	where $\lambda $ is the Perron value of the adjacency matrix. Moreover, $ \mathcal{Y}(s_0,x_{1:n}) $ and $ \mathcal{Z}(s_0,n) $ are the possible output and noise values for a given initial state $ s_0 $ and input sequence $ x_{1:n} $.
\end{lem}
\begin{proof}
	The output sequence, $ y_{1:n} $, is a function of input sequence, $ x_{1:n} $, and channel noise, $ z_{1:n} $, which can be represented as the following
	\begin{align}
	y_{1:n} = x_{1:n} \oplus z_{1:n},  \label{adnoise}
	\end{align}
	where $ z_{1:n} \in \mathcal{Z}(s_0,n) $. The set of all output sequences $ \mathcal{Y}(s_0,x_{1:n}) $ can be obtained as $\mathcal{Y}(s_0,x_{1:n})=\{ x_{1:n} \oplus z_{1:n}| z_{1:n} \in \mathcal{Z}(s_0,n)\}$.
	Since for given $ x_{1:n} $, \eqref{adnoise} is bijective, we have the following
	\begin{align}
	|\mathcal{Y}(s_0,x_{1:n})|=| \mathcal{Z}(s_0,n)|. \label{yvsize}
	\end{align}
	
	For a given initial state $s_0\in\mathcal{S}$, define the binary indicator vector $ \zeta \in\{0,1\}^{|\mathcal{S}|} $ consisting of all zeros except for a 1 in the position corresponding to $s_0$; e.g. in Fig.\ref{fig:examp}, if starting from state $ S=0 $, then $\zeta  =[1, 0, 0]$. Observe that since each output of the finite-state additive channel triggers a different state transition, each sequence of state transitions has a one-to-one correspondence to the output sequence, given the input sequence.
	
	The total number of state trajectories after $n$-step starting from state $s_i$ is equal to sum of $i$-th row of $\mathcal{A}^n$ \cite{lind1995introduction}. Hence, because of a one-to-one correspondence between state sequences and output sequences then $|\mathcal{Z}(s_0,n)|=\zeta ^\top \mathcal{A}^n \mathbbm{1}$. 
	
	Next, we show the upper and lower bounds in \eqref{outlam}. According to the Perron-Frobenius Theorem, for an irreducible $|\mathcal{S}|\times |\mathcal{S}|$ matrix $\mathcal{A} $ (or, equivalently, the adjacency matrix for a strongly connected graph), the entries of eigenvector $v \in \mathbb{R}^{|\mathcal{S}|}$ corresponding to $\lambda$ are strictly positive \cite[Thm. 8.8.1]{godsil2001strongly},\cite[Thm. 4.2.3]{lind1995introduction}. Therefore, multiplying $ \mathcal{A} $ by $ \mathcal{A} v=\lambda v $ results in $ \mathcal{A}^n v=\lambda^n v $ for $ n \in \mathbb{N} $.
	Left multiplication by the indicator vector, $ \zeta ^\top $ yields
	\begin{align}
	\zeta ^\top\mathcal{A}^n v = \lambda^n \zeta ^\top v. \label{eigv}
	\end{align}
	Denote minimum and maximum element of vector $ v $ by $ v_{min} $ and $ v_{max} $ respectively.	Hence, considering that all the elements in both sides of \eqref{eigv} are positive, we have
	\begin{align*}
	v_{min}\zeta ^\top\mathcal{A}^n \mathbbm{1}\leq \zeta ^\top\mathcal{A}^n v &\leq v_{max} \lambda^n \zeta ^\top \mathbbm{1} = v_{max} \lambda^n,
	\end{align*}
	where $ \mathbbm{1} $ is all-one column vector. Therefore, dividing by $v_{min}$, we have
	\begin{align}
	|\mathcal{Y}(s_0,x_{1:n})| =\zeta ^\top\mathcal{A}^n &\leq \frac{v_{max}}{v_{min}} \lambda^n = \beta \lambda^n, \label{out_up}
	\end{align}
	where $ \beta:= v_{max}/v_{min} > 0 $.	Moreover, for deriving the lower bound similar to above, we have
	\begin{align*}
	v_{min}\lambda^n \zeta ^\top \mathbbm{1}\leq \zeta ^\top\mathcal{A}^n v &\leq v_{max} \zeta ^\top\mathcal{A}^n \mathbbm{1}  =v_{max} |\mathcal{Y}(s_0,x_{1:n}) |.
	\end{align*}
	Let  $ \alpha:= v_{min}/v_{max}=1/\beta > 0 $, hence $\alpha \lambda^n \leq |\mathcal{Y}(s_0,x_{1:n}) |$ which combining it with \eqref{out_up} results in \eqref{outlam}.
\end{proof}
\subsection{Converse} \label{uppb} 
We prove no coding method can do better than \eqref{c0f}. 

Let $ m \in \mathcal{M} $ be the message to be sent and  $ y_{1:n} $ be the output sequence received such that 
\begin{align*}
y_i=f_{m,i}(y_{1:i-1})\oplus z_i, \, i=1,\dots, n,
\end{align*}
where $ z_{1:n} \in \mathcal{Z}(s_0,n) \in \mathcal{X}^{n}$ is the additive noise and $ f_{m,i} $ the encoding function. Therefore, the output is a function of encoding function and noise sequence, i.e., $ y_{1:n}=\psi(f_{m,1:n},z_{1:n}) $.
We denote all possible outputs $ \Psi (\mathcal{F_{\mathcal{M}}},\mathcal{Z}(s_0,n))=\{y_{1:n}| m \in \mathcal{M}, z_{1:n} \in \mathcal{Z}(s_0,n) \} $, where $ \mathcal{F_{\mathcal{M}}}=\{f_{m,t}: m \in \mathcal{M}\} $ is the family of encoding functions. 

For having a zero-error code any two $ m, m' \in \mathcal{M}, m\neq m'$ and any two $ z_{1:n}, z_{1:n}' \in \mathcal{Z}(s_0,n)$ must result in $ \psi(f_{m,1:n},z_{1:n})\neq \psi(f_{m',1:n},z_{1:n}') $. Note that when $ m = m' $, (even with feedback) at first position that $ z_{1:n} \neq z_{1:n}' $ will result in $ \psi(f_{m,1:n},z_{1:n})\neq \psi(f_{m,1:n},z_{1:n}') $. Therefore, assuming the initial condition is known at both encoder and decoder,
\begin{align*}
|\Psi (\mathcal{F_{\mathcal{M}}},\mathcal{Z}(s_0,n))| =M|\mathcal{Z}(s_0,n)| \leq q^n.
\end{align*}
Therefore, $ M $ is an upper bound on the number of messages that can be transmitted when initial condition is not available. We know that 
$ \alpha \lambda^n| \leq |\mathcal{Z}(s_0,n)| \leq \beta \lambda^n $. Therefore,
\begin{align*}
C_{0f} \leq \sup_{f \in F_{\mathcal{M}}} \frac{\log M}{n}  &\leq \frac{1}{n} \log \frac{q^n}{\alpha \lambda^n} =\log q-\log \lambda -\frac{\log\alpha}{n}  .
\end{align*}
Moreover, $ \lim_{n \rightarrow \infty} \frac{1}{n} \log \alpha = 0 $, which proves the converse in \eqref{c0f}.
\subsection{Achievability} \label{downb}

A coding method is proposed that achieves \eqref{c0f}. Consider a code of length $ n $ such that first $ k<n $ symbols are the data to be transmitted and the rest of $ n-k $ symbols serve as parity check symbols. 

We know that for an input of size $ k $ there are $ |\mathcal{Y}(k)|=|\mathcal{Z}(k)|=|\cup_{s_0} \mathcal{Z}(s_0,k)|$ possible output sequences, which is bounded as follows
\begin{align*}
\alpha \lambda^k \leq |\mathcal{Z}(s_0,k)| \leq |\mathcal{Y}(k)| \leq (|\mathcal{S}| \beta) \lambda^k.
\end{align*}
The transmitter having the output sequence $ y_{0:k-1} $, sends the receiver which output pattern (e.g. a message from $\{1,\dots, |\mathcal{Y}(k)|\} $ ) was received using the $ n-k $ parity check symbols. Assume that the transmitter sends the parity check symbols with a rate slightly below the zero-error feedback capacity, i.e., $ R=C_{0f}-\delta $, where $ \delta >0 $ is arbitrary small.\footnote{The reason to choose $ \delta $ is to deal with situation when $ C_{0f} $ is achieved when blocklength tends to infinity.}
Therefore, \[C_{0f}-\delta. = \frac{\log |\mathcal{Y}(k)|}{n-k}. \]

Using the upper bound on size of the output, i.e., $ |\mathcal{Y}(k)| \leq (|\mathcal{S}| \beta)\lambda^k $ and rearranging the inequality, gives
\begin{align*}
k &\geq \frac{(C_{0f}-\delta)n-\log(|\mathcal{S}|\beta)}{(C_{0f}-\delta +\log \lambda)}.
\end{align*}

Considering the fact that the total rate of coding is upper-bounded by $ C_{0f} $, we have
\begin{align*}
C_{0f} &\geq \frac{k}{n}\log q \geq \frac{C_{0f}-\delta-\frac{\log(|\mathcal{S}|\beta)}{n}}{C_{0f}-\delta +\log \lambda} \log q.
\end{align*}

Rearranging gives the following.
\begin{align*}
C_{0f} & \geq  \log q-\log \lambda  -\delta \bigg(1-\frac{1}{C_{0f}}\bigg)\log q -\frac{\log(|\mathcal{S}|\beta)}{nC_{0f}} \log q.
\end{align*}

By choosing $ \delta $ small and making $ n $ large, the last two terms disappear and this concludes the proof.

\section{Proof of the Zero-Error Capacity Lower Bound} \label{sec:proofl}
First, we give the following Lemma.
\begin{lem} \label{lemnss}
	Let $ \mathcal{G}(s_0,y_{1:n}) $ be subset of the inputs that can result in output $y_{1:n}$ with initial state $ s_0 $ for the finite-state additive noise channel. The following holds
	\begin{align}
	\alpha \lambda^n \leq |\mathcal{G}(s_0,y_{1:n})| \leq \beta \lambda^n,
	\end{align}
\end{lem}
where $ \alpha $ and $ \beta $ are constants appeared in \eqref{outlam}.
\begin{proof}
	The subset of the inputs that can result in output $y_{1:n}$ with initial state $ s_0 $, $ \mathcal{G}(s_0,y_{1:n}) $ is defined as the following
	\begin{align*}
	\mathcal{G}(s_0,y_{1:n})=\{ x_{1:n} |x_{1:n} \oplus z_{1:n} = y_{1:n}, z_{1:n} \in \mathcal{Z}(s_0,n)\}.
	\end{align*}
	Fixing $ y_{1:n} $, the mapping $x_{1:n} \to z_{1:n}$ in \eqref{adnoise} is bijective, hence $|\mathcal{G}(s_0,y_{1:n})|=| \mathcal{Z}(s_0,n)|$. Combining it with \eqref{yvsize} yields $ |\mathcal{G}(s_0,y_{1:n})|= |\mathcal{Y}(s_0,x_{1:n})| $. Moreover, Lemma \ref{lem:out} gives the bounds on $ |\mathcal{Y}(s_0,x_{1:n})| $. 
\end{proof}
\begin{figure}[t]
	\centering
	\begin{tikzpicture}
	[->, auto, semithick, node distance=3.4cm]
	\tikzstyle{every state}=[fill=white,draw=black,thick,text=black,scale=1]
	\node[state]   (S1)       {$ S_i=0 $};
	\node[state]   (S2)[right of=S1]   {$ S_i=1 $};
	\path
	(S1) edge[out=160, in=200,looseness=5] node [left=0.1cm] {$ Z_i=0 $}(S1)
	edge[bend left] node [above] {$ Z_i=1 $}  (S2)
	(S2) edge[bend left]  node [below] {$ Z_i=0 $}  (S1);
	\end{tikzpicture}
	\caption{State transition diagram of a noise process in a channel at which no two consecutive errors can happen in the channel.}
	\label{fig:ex2stat}
\end{figure}
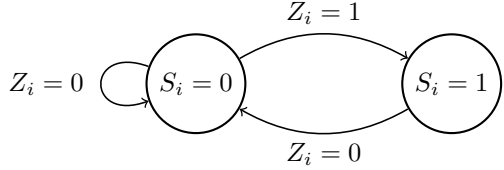
Let $ c(1) \in \mathcal{X}^n $ be the first codeword for which adjacent inputs denoted by $ \mathcal{Q}(c(1)) $. Again, each output sequence is in $\mathcal{Y}_T(c(1)):= \cup_{s_0 \in \mathcal{S}} \mathcal{Y}(s_0,c(1)) $. Hence,
\begin{align}
\mathcal{Q}(c(1))&=\bigcup_{y_{1:n} \in \mathcal{Y}_T(c(1)) } \mathcal{G}(y_{1:n}), \label{q3}
\end{align}
where, $ \mathcal{G}(y_{1:n}):=\bigcup_{s_0 \in \mathcal{S}} \mathcal{G}(s_0,y_{1:n}) $,
which gives 
\begin{align*}
|\mathcal{Q}(c(1))|&\leq \sum_{y_{1:n} \in \mathcal{Y}_T(c(1))}\sum_{s_0 \in \mathcal{S}} |\mathcal{G}(s_0,y_{1:n})|.
\end{align*}
Using Lemma \ref{lemnss}, we have
\begin{align*}
|\mathcal{Q}(c(1))|&\leq|\cup_{s_0 \in \mathcal{S}} \mathcal{Y}(s_0,c(1))|(|\mathcal{S}| \times\beta \lambda^n).
\end{align*}
According to \eqref{outlam}, for any initial state the number of outputs is upper-bounded by $ \beta \lambda^n $. Therefore,
\begin{align*}
|\mathcal{Q}(c(1))|&\leq \big(|\mathcal{S}| (\beta \lambda^n) \big)\times(\beta|\mathcal{S}| \lambda^n) =\big( \beta|\mathcal{S}| \lambda^n \big)^2.
\end{align*}
By choosing non-adjacent inputs as the codebook, results in an error-free transmission. The above argument is true for other codewords, i.e.,
\begin{align*}
|\mathcal{Q}(c(i))|&\leq \big( \beta|\mathcal{S}| \lambda^n \big)^2, i \in \{1,\dots, M \}, 
\end{align*}
where $ M $ is the number of codewords in the codebook such that union of corresponding $\mathcal{Q}(c(i))$ for $i=1,\dots,M,$ covers $ \mathcal{X}^n $. Then,
\begin{align*}
q^n = |\bigcup_{i=1}^M \mathcal{Q} (c(i))| &\leq \sum_{i=1}^{M} |\mathcal{Q} (c(i))|\leq M 	\times \big( \beta|\mathcal{S}| \lambda^n \big)^2.		
\end{align*}
A a result, the number of distinguishable inputs is lower bounded by $M \geq q^n/( \beta|\mathcal{S}| \lambda^n )^2$.
Therefore, according to zero-error capacity definition
\begin{align*}
C_0&\geq \log \frac{q}{(\beta |\mathcal{S}|)^2 (\lambda)^{2n}} =\log q-2\log \lambda - \frac{2}{n} \log (\beta |\mathcal{S}|).
\end{align*}
If $n$ is large, the last term vanishes and proves the lower bound in \eqref{c0l}.
\section{Examples} \label{sec:examps}
Here, we provide some examples, and for them, compute $ C_{0f} $ explicitly. Examples \ref{2stat} and \ref{3stat} consider channels with isolated and limited runs of errors. In Example \ref{estat} we consider a Gilbert-Elliot channel. Moreover, for examples \ref{2stat} and \ref{3stat}, we investigate that minimum value of ordinary feedback capacity $ C_f $ over the transition probabilities and observe how far is this natural upper bound from the zero-error feedback capacity.\smallskip

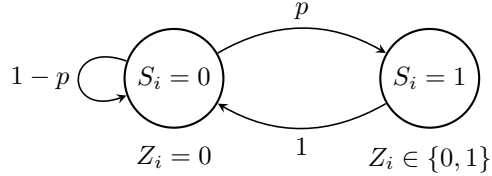
\begin{figure}[t]
	\centering
	\begin{tikzpicture}
	[->, >=stealth, auto, semithick, node distance=3.4cm]
	\tikzstyle{every state}=[fill=white,draw=black,thick,text=black,scale=1]
	\node[state] (A1)   {$ S_i=0 $};
	\node[] (state0)[below=0.1cm of A1] {$ Z_i=0 $};
	\node[state]   (A2)[right of=A1]   {$S_i=1 $};
	\node[] (state1)[below=0.1cm of A2] {$  Z_i \in \{0,1\} $};
	\path
	(A1) edge[out=160, in=200,looseness=5] node [left] {$ 1-p $} (A1)
	edge [bend left] node [above] {$ p $}  (A2)
	(A2) edge[bend left]  node [below] {$ 1$}  (A1);
\end{tikzpicture}
\caption{Markov chain for channel states in Example \ref{estat}.}
\label{fig:estat}
\end{figure}

\begin{ex}\label{2stat}
	Consider a channel with no two consecutive errors (Fig.~\ref{fig:ex2stat}). If $ q = 2 $ then $ C_0=C_{0f}=0 $. Whilst, if $ q \geq 3 $ it has a zero-error feedback capacity of $ C_{0f}=\log q-\log(\frac{1+\sqrt{5}}{2}) $ bit/use where $ \frac{1+\sqrt{5}}{2} $ is known as the golden ratio.
	
	Moreover, assuming Markovianity with the transition probability $ P(S_{i+1}=1|S_i=0)=p $, the ordinary feedback capacity is $ C_f (p)= \log q-\frac{H(p)}{1+p}$ from \eqref{ccf}, where $ H(.) $ is the binary entropy function. It turns out that 
	$ C_{0f}= \min_{p\in (0,1)} C_f(p) $.
\end{ex} \smallskip

\begin{ex}\label{3stat}
	 The example of Fig.~\ref{fig:examp} represents a channel with no more than two consecutive errors, having adjacency matrix
\begin{align*}
\mathcal{A}=\begin{bmatrix}
1&1&0\\
1&0&1\\
1&0&0
\end{bmatrix}.
\end{align*}
	 If $ q = 2 $ then $ C_0=C_{0f}=0 $ and if $ q =3 $ it has $ C_{0f} = 0.7058$. 
	 
	 If the channel states are Markov with transition probabilities $ P(S_{i+1}=1|S_i=0)=p $ and $ P(S_{i+1}=2|S_i=1)=r $, it can be shown that\\ $ \min_{p,r\in (0,1)}C_f (p,r)= 0.7935 > C_{0f} $.
\end{ex}\smallskip

\begin{ex}\label{estat} Consider a Gilbert-Elliot channel with input alphabet of size $ q=5 $ and two states (Fig.~\ref{fig:estat}). When the state $ S_i=0 $ the channel is error-free, i.e., $ P(Z_i\neq 0|S_i=0)= 0 $ and when state $ S_i=1 $ it acts like a noisy type-writer channel (Fig.~\ref{fig:pen}) which is also known as the Pentagon channel \cite{shannon1956zero}. In this state, the probability of error for any input symbol is $ P(Z_i=1|S_i=1)= r $ and thus the probability of error-free transmission is $ P(Z_i=0|S_i=1)=1-r $. Figure~\ref{fig:estat} shows this channel's state transition diagram. However, this channel does not fit Definition \ref{def:fanchannel}, because outgoing edges are not associated with unique noise values. This reflects the fact that the noise process is a hidden Markov model, not a Markov chain, and the same state sequence can yield multiple noise sequences. 
\end{ex} 	
Nonetheless, in the following we show an equivalent representation of this channel compatible with Definition \ref{def:fanchannel}. The resultant model (shown in Fig.~\ref{fig:eq2}) is a state machine that produces the same set of noise sequences, where the edges define the noise values in each transmission.
\begin{figure}
	\centering
	\begin{tikzpicture}[->, >=stealth, auto, semithick, node distance=1cm]
	\tikzstyle{every state}=[fill=white,thick,text=black,scale=1]
	\node[]    (S1)  				{0};
	\node[]    (S2)[below of=S1]   {1};
	\node[]    (S3)[below of=S2]   {2};
	\node[]    (S4)[below of=S3]   {3};
	\node[]    (S5)[below of=S4]   {4};
	\node[]    (S6)[right of=S1,right=2cm]   {0};
	\node[]    (S7)[below of=S6]   {1};
	\node[]    (S8)[below of=S7]   {2};
	\node[]    (S9)[below of=S8]   {3};
	\node[]    (S10)[below of=S9]   {4};
	
	\path
	(S1.0) edge   (S6)
	edge    (S7.180)
	(S2.0) edge  (S7)
	edge  (S8.180)
	(S3.0) edge   (S8)
	edge    (S9.180)
	(S4.0) edge  (S9)
	edge  (S10.180)
	(S5.0) edge   (S10)
	edge    (S6.180);
	\end{tikzpicture}
	\caption{Pentagon channel.}
	\label{fig:pen}
\end{figure}
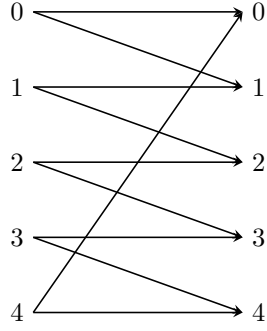
Note that if the channel is in state $ S_i=0 $, the noise can only take value $ Z_i=0 $, but in state $ S_i=1 $, the noise $ Z_i\in \{0,1\} $, thus $ Z_i\in \{0,1\} $ at all times. In the sequel, we show that 
\begin{align}
P(Z_{i+1}=1|Z_i=1,z_{1:i-1})&=0,\label{zz11}\\
P(Z_{i+1}=0|Z_i=0,z_{1:i-1})&>0,\label{zz00}\\
P(Z_{i+1}=1|Z_i=0,z_{1:i-1})&>0,\label{zz10}
\end{align}
whenever the conditioning sequence of $ Z_i=j,Z_{1:i-1}=z_{1:i-1}, j\in \{0,1\} $ occurs with non-zero probability. Therefore, irrespective of past noises the state machine shown in Fig.~\ref{fig:eq2} can produce all noise sequences that occur with nonzero probability. It should be stressed that this noise process may not be a stochastic Markov chain, however, it is a topological Markov chain \cite[Ch.2]{lind1995introduction}. 
First, note by inspection of Fig.~\ref{fig:estat} that the noise process has zero probability of taking value $ 1 $ twice in a row. Thus 
$ P(Z_{i+1}=1, Z_i =1, z_{1:i-1}) = 0 $.
Using Bayes rule, it then follows that
\begin{align*}
P(Z_{i+1}=1|Z_i=1,z_{1:i-1}) &=0,
\end{align*}
whenever $ P(Z_i=1, z_{1:i-1}) >0 $.

Next we show \eqref{zz00}-\eqref{zz10}. Let $z_{1:i-1}$ be any past noise sequence such that $P(Z_i=0, z_{1:i-1}) >0$.Therefore, $\exists s_{1:i} $ such that 
\begin{align}
P(Z_i=0,z_{1:i-1},s_{1:i}) &=P(Z_i=0|s_i)P(z_{1:i-1},s_{1:i})>0.\label{con0}
\end{align}
From Fig.~\ref{fig:estat}, $ P(Z_{i+1}=0,Z_i=0|s_i=j)>0, j\in \{0,1\} $. Thus
\begin{align*}
P(Z_{i+1}=0,Z_i=0,z_{1:i-1},s_{1:i})&=P(Z_{i+1}=0,Z_i=0|s_i)P(z_{1:i-1},s_{1:i})>0,
\end{align*}
since the second factor on the RHS is positive, by \eqref{con0}. Therefore, $ P(Z_{i+1}=1, Z_i=0, z_{1:i-1}) > 0 $, and \eqref{zz00} holds. Now, we show \eqref{zz10}. If $ S_i=0 $, it can be shown from Fig.~\ref{fig:estat} and the noise probabilities that 
\begin{align}
P(Z_{i+1}=1,Z_i=0|S_i=0)&=rp>0. \label{z10rp}
\end{align}
Therefore, 
\begin{align*}
P(Z_{i+1}=1,Z_i=0,z_{1:i-1}) &\geq P(Z_{i+1}=1,Z_i=0,z_{1:i-1},S_i=0, s_{1:i-1})\\
&=P(Z_{i+1}=1,Z_i=0|S_i=0)P(S_i=0,z_{1:i-1},s_{1:i-1})\\
&=rp\,P(S_i=0,z_{1:i-1},s_{1:i-1})>0.
\end{align*}
Note from Fig.~\ref{fig:estat} that $ P(S_i=0|s_{i-1})>0 $. Thus,
\begin{align*}
P(S_i=0,z_{1:i-1},s_{1:i-1})&=P(S_i=0|s_{i-1})P(z_{1:i-1},s_{1:i-1})>0.
\end{align*}
Consequently, \eqref{zz11}-\eqref{zz10} hold yielding the state machine in Fig.~\ref{fig:eq2}. Note that, $ \hat{S}_i=0 $ corresponds to $ Z_i=0 $ and  $ \hat{S}_i=1 $, to $ Z_i=1 $.

\begin{figure}[t]
	\centering
	\begin{tikzpicture}
	[->, auto, semithick, node distance=3.4cm]
	\tikzstyle{every state}=[fill=white,draw=black,thick,text=black,scale=1]
	\node[state]   (S1)       {$ \hat{S}_i=0 $};
	\node[state]   (S2)[right of=S1]   {$ \hat{S}_i=1 $};
	\path
	(S1) edge[out=160, in=200,looseness=5] node [left=0.1cm] {$ Z_{i}=0 $} (S1)
	edge[bend left]  node [above] {$ Z_{i}=1 $}  (S2)
	(S2) edge[bend left]  node [below] {$Z_{i}=0 $}  (S1);
	\end{tikzpicture}
	\caption{State machine generating the noise sequence of Example \ref{estat}.}
	\label{fig:eq2}
\end{figure}
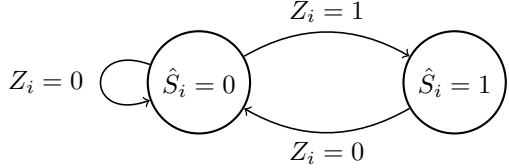
Now, we can use the results of Theorem \ref{thm:zero-error}, to get
 \[ \log5-2\log\bigg(\frac{1+\sqrt{5}}{2}\bigg) \leq C_0 \leq C_{0f} = \log5-\log\bigg(\frac{1+\sqrt{5}}{2}\bigg). \] 
 This shows that the zero-error feedback capacity of some channels with different structure than Definition \ref{def:fanchannel}, such as time-varying state transmissions (non-homogeneous Markov chains) and even transitions that depend on previous transmissions can be explicitly obtained.
	\section{Conclusion} \label{sec:conclusions}
	We introduced a formula for computing the zero-error feedback capacity for a class of additive noise channels without state information at the decoder and encoder. This reveals a close connection between the topological entropy of the underlying noise process and the zero-error communication.
	Moreover, a lower bound on zero-error capacity (without feedback) was given based on the topological entropy. 
	
	Future work includes extending these results to a more general class of channels.
\end{spacing}
\bibliographystyle{IEEEtran}
\bibliography{ieeetranb}
\end{document}